\documentclass[reprint, superscriptaddress, amsmath, amssymb, aps, twocolumn, floatfix]{revtex4-1}

\usepackage[utf8]{inputenc}
\usepackage[T1]{fontenc}
\usepackage{hyperref}
\usepackage{graphicx}
\usepackage[caption=false]{subfig}
\usepackage{amsthm}
\usepackage{bm}
\usepackage{tensor}
\usepackage{diagbox}
\usepackage{booktabs}
\usepackage{xcolor}

\newcommand{\sket}[1]{|{#1}\rangle}

\newcommand{\sbraket}[2]{\langle{#1}|{#2}\rangle}

\newcommand{\Dff}{\mathrm{d}}
\newcommand{\opnorm}[1]{\left\|#1\right\|_{\mathrm{op}}}
\newcommand{\HSnorm}[1]{\left\|#1\right\|_{\mathrm{HS}}}
\newcommand{\barth}{\bar{\theta}}

\newtheorem{thm}{Theorem}

\newtheorem{cor}{Corollary}
\theoremstyle{definition}
\newtheorem{definition}{Definition}

\begin{document}

\title{Confidence uncertainty: position and momentum can be jointly determined
with a guaranteed probability}

\author{Jia-Yi Lin}
    \affiliation{National Laboratory of Solid State Microstructures and School of Physics,  Collaborative Innovation Center of Advanced Microstructures, Nanjing University, Nanjing 210093, China.}
\author{Xin-Yu Li}
    \affiliation{Institute for Brain Sciences and Kuang Yaming Honors School, Nanjing University, Nanjing 210023, China.}
\author{Wei Wang}
    \email{wangwei@nju.edu.cn}
    \affiliation{National Laboratory of Solid State Microstructures and School of Physics,  Collaborative Innovation Center of Advanced Microstructures, Nanjing University, Nanjing 210093, China.}
    \affiliation{Institute for Brain Sciences and Kuang Yaming Honors School, Nanjing University, Nanjing 210023, China.}
\author{Shengjun Wu}
    \email{sjwu@nju.edu.cn}
    \affiliation{National Laboratory of Solid State Microstructures and School of Physics,  Collaborative Innovation Center of Advanced Microstructures, Nanjing University, Nanjing 210093, China.}
    \affiliation{Institute for Brain Sciences and Kuang Yaming Honors School, Nanjing University, Nanjing 210023, China.}

\date{\today}

\begin{abstract}
Standard-deviation and entropic formulations of uncertainty principle
capture the spread of the probability distribution but
say little about the probability itself contained in a small region. We
introduce the \emph{confidence uncertainty} $\Delta^{c}x(\theta_x)$ as the
minimal Lebesgue measure of the support set in which the particle is found
with probability at least $\theta_x$, and the companion \emph{interval
confidence uncertainty} $\Delta^{I}x(\theta_x)$ which restricts the support to
a single interval. We prove two complementary uncertainty inequalities. (i)
For $\theta_x+\theta_p\le 1$ both confidence uncertainties can be made
arbitrarily small simultaneously, so that no nontrivial product bound holds;
in particular, position and momentum can be \emph{jointly localised with
probability at least~$50\%$}. (ii) For $\theta_x+\theta_p>1$ a lower
bound holds: combining Lenard's projection inequality with the Donoho--Stark
operator-norm bound we obtain $\Delta^{c}x\,\Delta^{c}p\geq
2\pi\hbar\bigl(\sqrt{\theta_x\theta_p}-\sqrt{(1-\theta_x)(1-\theta_p)}\bigr)^{\!2}$,
and for the interval version we obtain the sharp implicit Landau--Pollak bound
$\Delta^{I}x\,\Delta^{I}p\geq 4\hbar\,\lambda_{0}^{-1}\!\bigl((\sqrt{\theta_x\theta_p}-\sqrt{(1-\theta_x)(1-\theta_p)})^{2}\bigr)$,
where $\lambda_{0}(c)$ is the largest prolate-spheroidal eigenvalue. We
support the analytical bounds with numerical evaluation of $\lambda_{0}(c)$,
provide closed-form small-$c$ and large-$c$ asymptotics, compute the optimal
Slepian-superposition states that saturate the interval bound, and compare
the resulting product against the variance Heisenberg--Kennard, the
Bia\l{}ynicki-Birula--Mycielski entropic, and the Donoho--Stark
concentration bounds. The unified picture provides a complete phase diagram on
$(\theta_x,\theta_p)\in[0,1]^{2}$.
\end{abstract}

\maketitle

\section{Introduction}\label{sec:intro}

Quantum mechanics is the cornerstone of modern physics, and the uncertainty
principle is one of its most counter-intuitive features. The original
formulation by Heisenberg~\cite{Heisenberg1927} and the rigorous proof by
Kennard~\cite{Kennard1927} state that the standard deviations of the
position and momentum of a particle satisfy
\begin{equation}
\Delta x\cdot\Delta p\geq\frac{\hbar}{2}.\label{eq:HK}
\end{equation}
Robertson~\cite{Robertson1929} and Schr\"odinger~\cite{Schrodinger1930}
extended this to arbitrary observables, and a vast literature has since
explored the implications of \eqref{eq:HK} for entanglement
detection~\cite{Guhne2004,HofmannTakeuchi2003}, quantum
nonlocality~\cite{Oppenheim2010,LiDuQiao2014},
and improved bounds for mixed states~\cite{Park2005,Heydari2015,MacconePati2014}.

The variance-based formulation has well-known limitations. A symmetric
bimodal distribution with two narrow lobes carries a large $\Delta x$ and yet
the particle is sharply localised inside each lobe. From an
information-theoretic standpoint the relevant quantity is not the spread of
the distribution but the probability concentrated in a small set, and this
motivated the entropic uncertainty relations: Hirschman--Beckner~\cite{Hirschman1957,Beckner1975},
Bia\l{}ynicki-Birula and Mycielski~\cite{BialynickiBirula1975}, Deutsch~\cite{Deutsch1983},
Maassen and Uffink~\cite{Maassen1988}, Wu-Yu-Mølmer~\cite{Wu2009}, and the comprehensive recent review of
Coles \emph{et al.}~\cite{Coles2017}.
Uncertainty equality with quantum memory~\cite{Wang2019} was found, and various extensions of entropic uncertainty relation
were given very recently~\cite{Huang2021,Huang2024R,Zhao2024,Huang2024A}.
Entropic uncertainty has found use in
quantum cryptography~\cite{Berta2010} and in various entropic
bounds~\cite{Renyi1961,Pegg1998,Wehner2010,Dodonov2015,Zhou2016,Zhang2016,Rastegin2019}.

Even entropic uncertainty does not directly answer the operational question:
\emph{within how short an interval can the position be located with
probability at least $\theta_x$, while the momentum is simultaneously
located within probability at least $\theta_p$?} For
$\theta_x=\theta_p=1$ this question is the classical Paley--Wiener obstacle
(a function compactly supported in both position and momentum vanishes), but
the partial-confidence regime $\theta_x,\theta_p<1$ is more nuanced. A close
classical analogue exists in signal processing through the Slepian--Pollak--Landau
theory of prolate-spheroidal wave functions and time-frequency
concentration~\cite{SlepianPollak1961,LandauPollak1961,LandauPollak1962,Slepian1983}
and in the Donoho--Stark uncertainty principle for measurable
sets~\cite{DonohoStark1989,Lenard1972,FollandSitaram1997,Nazarov2008}.
Our paper translates these tools to the quantum-mechanical position--momentum
setting, defines the corresponding \emph{confidence
uncertainty}, and gives lower bounds across the full parameter range.

\paragraph*{Contributions.}
\begin{enumerate}\itemsep0pt
\item \emph{Definitions} (Sec.~\ref{sec:defs}): the confidence uncertainty
  $\Delta^{c}x(\theta_x)$ and the interval confidence uncertainty
  $\Delta^{I}x(\theta_x)$.
\item \emph{Theorem~\ref{thm:half}: a $50\%$ joint-localisation result.}
  When $\theta_x+\theta_p\le 1$ no positive lower bound exists on
  $\Delta^{c}x\,\Delta^{c}p$. In particular, position and momentum can be
  simultaneously localised to arbitrary precision with confidence
  $\theta_x=\theta_p=\tfrac12$.
\item \emph{Theorem~\ref{thm:LP}: tight Landau--Pollak lower bound.} For
  $\theta_x+\theta_p>1$,
  $\Delta^{c}x\,\Delta^{c}p\ge 2\pi\hbar(\sqrt{\theta_x\theta_p}-\sqrt{(1-\theta_x)(1-\theta_p)})^{2}$,
  and for the interval version the bound is implicit through the largest
  prolate-spheroidal eigenvalue $\lambda_0(c)$. The interval bound is sharp
  (saturated by Slepian-superposition states).
\item \emph{Numerical landscape} (Sec.~\ref{sec:numerics}): explicit
  evaluation of $\lambda_0(c)$, asymptotic forms, the saturating wave
  functions, and numerical maps of the lower bound on $(\theta_x,\theta_p)\in[0,1]^{2}$.
\item \emph{Comparison with established uncertainty principles}
  (Sec.~\ref{sec:compare}): we exhibit the regime in which our bound is
  tighter than the variance Heisenberg--Kennard, the
  Bia\l{}ynicki-Birula--Mycielski entropic, and the Donoho--Stark sum-form
  bound, and we show that minimum-variance Gaussian states are not optimal
  for confidence uncertainty.
\end{enumerate}

\paragraph*{Roadmap.}
Sections \ref{sec:defs}--\ref{sec:interval} formulate and prove our two
uncertainty relations. Section~\ref{sec:numerics} develops the numerical and
asymptotic picture. Section~\ref{sec:compare} situates the new bounds within
the existing uncertainty-relation literature. Section~\ref{sec:conclu}
concludes. Appendix~\ref{app:proofs} collects the proofs.

\section{Confidence uncertainty}\label{sec:defs}

For a particle in a pure state $\sket{\psi}$, the probability of finding the
particle in a position interval $[a,b]$ equals
$\int_a^b|\sbraket{x}{\psi}|^2\,\Dff x$. We are interested in the
\emph{minimal} support of position values that captures at least a chosen
fraction~$\theta_x$ of the total probability.

\begin{definition}[Confidence uncertainty]
The \emph{confidence uncertainty} $\Delta^{c}x(\theta_x)$ of position with
confidence level $\theta_x\in[0,1]$ is
\begin{equation}
\Delta^{c}x(\theta_x)=\inf\!\left\{\mu(X):\int_{X}|\sbraket{x}{\psi}|^{2}\,\Dff x\ge\theta_x\right\},
\label{eq:def_cx}
\end{equation}
where $X\subseteq\mathbb{R}$ ranges over Lebesgue-measurable sets and
$\mu(\cdot)$ is the Lebesgue measure. It is also written as
$\Delta^{c}x(\theta_x,\sket{\psi})$ when the dependence on the state must be
explicit. The momentum confidence uncertainty
$\Delta^{c}p(\theta_p,\sket{\psi})$ is defined analogously through the
momentum-space wave function
$\sbraket{p}{\psi}=\frac{1}{\sqrt{2\pi\hbar}}\int e^{-ipx/\hbar}\sbraket{x}{\psi}\,\Dff x$.
\end{definition}

\begin{definition}[Interval confidence uncertainty]
The \emph{interval confidence uncertainty}
$\Delta^{I}x(\theta_x)$ restricts the admissible support to a single
interval:
\begin{equation}
\Delta^{I}x(\theta_x)=\inf\!\left\{x_2-x_1:\int_{x_1}^{x_2}|\sbraket{x}{\psi}|^{2}\,\Dff x\ge\theta_x\right\}.
\label{eq:def_Ix}
\end{equation}
The momentum version $\Delta^{I}p(\theta_p)$ is analogous.
\end{definition}

By construction $\Delta^{c}x\le\Delta^{I}x$ and both are non-decreasing in
$\theta_x$. We seek inequalities of the form
\begin{equation}
\Delta^{\bullet}x(\theta_x)\cdot\Delta^{\bullet}p(\theta_p)\ge g\,\hbar,
\qquad \bullet\in\{c,I\},\label{eq:cur}
\end{equation}
with a non-negative constant $g=g(\theta_x,\theta_p)$ that is independent of
$\sket{\psi}$.

\section{Two regimes for confidence uncertainty}\label{sec:reg}

The $(\theta_x,\theta_p)$ plane splits into two regions in which the answer
to \eqref{eq:cur} is qualitatively different.
Figure~\ref{fig:phase} summarises the situation, which we now establish.

\begin{figure}[!htbp]
\centering
\includegraphics[width=0.9\columnwidth]{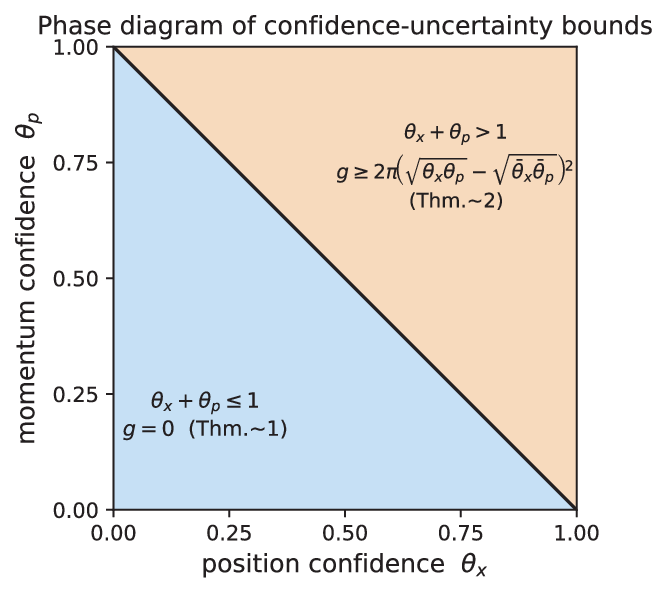}
\caption{Phase diagram for the confidence-uncertainty bound in the
$(\theta_x,\theta_p)$ plane. In the lower triangle $\theta_x+\theta_p\le 1$
no positive lower bound on $\Delta^{c}x\,\Delta^{c}p$ exists
(Theorem~\ref{thm:half}); in particular both uncertainties can be made
arbitrarily small simultaneously. In the upper triangle
$\theta_x+\theta_p>1$ Theorem~\ref{thm:LP} furnishes the strict positive
bound $g\ge 2\pi(\sqrt{\theta_x\theta_p}-\sqrt{\barth_x\barth_p})^{2}$ with
$\barth_\bullet\equiv 1-\theta_\bullet$.}
\label{fig:phase}
\end{figure}

\subsection{Below the diagonal: arbitrary joint localisation}\label{sec:below}

\begin{thm}\label{thm:half}
If $\theta_x+\theta_p\le 1$, then for every $\varepsilon>0$ there exists a
quantum state $\sket{\psi}$ such that
$\Delta^{c}x(\theta_x,\sket{\psi})<\varepsilon$ and
$\Delta^{c}p(\theta_p,\sket{\psi})<\varepsilon$ simultaneously. Equivalently
$g(\theta_x,\theta_p)=0$ in \eqref{eq:cur}. In particular,
$\theta_x=\theta_p=\tfrac{1}{2}$ is admissible: position and momentum can
be \emph{simultaneously localised to arbitrary precision with confidence at
least $50\%$}.
\end{thm}

\begin{proof}[Proof of Theorem~\ref{thm:half}]
We start with $\theta_x+\theta_p=1$ and write $\theta_x=P$, $\theta_p=1-P$.
Consider the family of states with parameters $L,W>0$ and $0\le P\le 1$,
\begin{align}
\psi(x;L,W,P)&=\frac{1}{C}\!\left[\sqrt{P}\,\psi_{\mathrm{rect}}(x;L)+\sqrt{1-P}\,\psi_{\mathrm{sinc}}(x;W)\right],
\nonumber\\
\psi_{\mathrm{rect}}(x;L)&=\frac{1}{\sqrt L}\,\mathrm{rect}(x/L),\nonumber\\
\psi_{\mathrm{sinc}}(x;W)&=\sqrt{\frac{W}{2\pi\hbar}}\,\mathrm{sinc}\!\left(\frac{Wx}{4\hbar}\right),
\label{eq:rect_sinc}
\end{align}
with normalisation
$C=\sqrt{1+4\sqrt{P(1-P)\frac{2\hbar}{\pi LW}}\,\mathrm{Si}\!\left(\frac{LW}{4\hbar}\right)}$
and $\mathrm{Si}(y)=\int_0^y\frac{\sin t}{t}\,\Dff t$. The Fourier transform
exchanges the rectangular and sinc parts:
\begin{align}
\phi(p;L,W,P)&=\frac{1}{C}\!\left[\sqrt{P}\,\phi_{\mathrm{sinc}}(p;L)+\sqrt{1-P}\,\phi_{\mathrm{rect}}(p;W)\right],
\nonumber\\
\phi_{\mathrm{sinc}}(p;L)&=\sqrt{\frac{L}{2\pi\hbar}}\,\mathrm{sinc}\!\left(\frac{Lp}{4\hbar}\right),\nonumber\\
\phi_{\mathrm{rect}}(p;W)&=\frac{1}{\sqrt W}\,\mathrm{rect}(p/W).\label{eq:rect_sinc_p}
\end{align}
The probability of $x\in[-L/2,L/2]$ for the state \eqref{eq:rect_sinc} is
\begin{align}
\int_{-L/2}^{L/2}|\psi|^{2}\Dff x
&=1-\frac{1-P}{C^{2}}\!\!\left[\int_{-\infty}^{-L/2}\!\!|\psi_{\mathrm{sinc}}|^{2}\Dff x+\int_{L/2}^{\infty}\!\!|\psi_{\mathrm{sinc}}|^{2}\Dff x\right]\nonumber\\
&> 1-\frac{1-P}{C^{2}}>1-(1-P)=P,\label{eq:probLb}
\end{align}
where the strict inequality uses $C>1$. Therefore
$\Delta^{c}x(\theta_x=P,\sket{\psi})<L$, and an analogous calculation in the
momentum representation yields
$\Delta^{c}p(\theta_p=1-P,\sket{\psi})<W$. Choosing $L,W\to0^{+}$ gives the
stated arbitrary joint localisation when $\theta_x+\theta_p=1$. The case
$\theta_x+\theta_p<1$ follows immediately because each confidence uncertainty
is monotone non-decreasing in its confidence level. The
$\theta_x=\theta_p=\tfrac12$ statement is the special case $P=\tfrac12$.
\end{proof}

\subsection{Above the diagonal: a sharp Landau--Pollak bound}\label{sec:above}

\begin{thm}[New lower bound on confidence uncertainty]\label{thm:LP}
For every quantum state and every $\theta_x,\theta_p\in(0,1]$ with
$\theta_x+\theta_p>1$,
\begin{equation}
\Delta^{c}x(\theta_x)\cdot\Delta^{c}p(\theta_p)\;\ge\;2\pi\hbar\,
\Bigl(\sqrt{\theta_x\theta_p}-\sqrt{(1-\theta_x)(1-\theta_p)}\Bigr)^{\!2}.
\label{eq:LP_bound}
\end{equation}
For the interval version, the analogous bound is sharp and implicit:
\begin{equation}
\lambda_{0}\!\left(\frac{\Delta^{I}x\cdot\Delta^{I}p}{4\hbar}\right)\ge
\Bigl(\sqrt{\theta_x\theta_p}-\sqrt{(1-\theta_x)(1-\theta_p)}\Bigr)^{\!2},
\label{eq:LP_bound_I}
\end{equation}
where $\lambda_0(c)$ is the largest eigenvalue of the prolate-spheroidal
integral operator $(K_c f)(u)=\int_{-1}^{1}\frac{\sin c(u-v)}{\pi(u-v)}f(v)\,\Dff v$
on $L^2[-1,1]$. Equivalently
\begin{equation}
\Delta^{I}x\cdot\Delta^{I}p\;\ge\;4\hbar\,\lambda_{0}^{-1}\!\Bigl((\sqrt{\theta_x\theta_p}-\sqrt{(1-\theta_x)(1-\theta_p)})^{2}\Bigr).
\label{eq:LP_bound_I_explicit}
\end{equation}
The bound \eqref{eq:LP_bound_I_explicit} is saturated by the leading Slepian
function (i.e.\ the principal eigenfunction of the prolate-spheroidal kernel).
\end{thm}

The proof, which combines Lenard's projection
inequality~\cite{Lenard1972} with the Donoho--Stark operator-norm
bound~\cite{DonohoStark1989}, is given in Appendix~\ref{app:LP}.

\begin{cor}\label{cor:logdiv}
For fixed $\theta_x=1$, the interval bound \eqref{eq:LP_bound_I_explicit}
diverges logarithmically as $\theta_p\to1$:
\begin{equation}
\Delta^{I}x\cdot\Delta^{I}p\;\sim\;-2\hbar\ln(1-\theta_p),\qquad\theta_p\to 1^{-}.
\label{eq:asymp_log}
\end{equation}
\end{cor}

\noindent
This logarithmic divergence is the quantum analogue of the asymptotic decay
$1-\lambda_0(c)\sim 4\sqrt{\pi c}\,e^{-2c}$ of the leading prolate-spheroidal
eigenvalue~\cite{Slepian1983}; see Sec.~\ref{sec:numerics}.

\section{Interval uncertainty: Slepian formulation}\label{sec:interval}

The interval uncertainty admits an exact reduction to a classical eigenvalue
problem in signal-processing theory~\cite{SlepianPollak1961,LandauPollak1961}.
Without loss of generality fix the candidate intervals
$X=[-L/2,L/2]$ and $P=[-W/2,W/2]$. Denote by $P_X$ the orthogonal projector
onto $L^{2}(X)$ and by $Q_P$ the orthogonal projector onto the band-limited
subspace $\mathcal{F}^{-1}\bigl(L^{2}(P)\bigr)$. The probability of finding
the position in $X$ given a band-limited state, or symmetrically the
probability of finding the momentum in $P$ given a position-truncated state,
is governed by
$\opnorm{P_X Q_P}^{2}=\lambda_{0}\!\left(\frac{LW}{4\hbar}\right)$. This is
the Slepian--Pollak--Landau identity~\cite{SlepianPollak1961}, and the
saturating eigenstates are the prolate-spheroidal wave functions
$\psi^{\mathrm{S}}_{0}$.

The ground-state wave function of $P_X Q_P P_X$ can be
written explicitly as a Fourier series confined to $[-L/2,L/2]$,
\begin{equation}
\psi(x)=\begin{cases}\sum_{n\in\mathbb{Z}}C_n\,L^{-1/2}e^{i 2\pi n x/L} & |x|\le L/2,\\
0 & |x|>L/2,\end{cases}
\label{eq:psi_series}
\end{equation}
with momentum-space wave function
\begin{equation}
\phi(p)=\sqrt{\frac{2L}{\pi\hbar}}\sin\!\frac{Lp}{2\hbar}\sum_{n}\frac{(-1)^{n}C_{n}}{Lp/\hbar-2n\pi}.
\label{eq:phi_series}
\end{equation}
The probability $P_W$ of the momentum lying in $[-W/2,W/2]$ takes the matrix
form $P_W=\frac{1}{\pi}C^{\dagger}A\,C$ with the symmetric matrix
\begin{equation}
A_{m,n}=(-1)^{m+n}\!\!\int_{-LW/(2\hbar)}^{LW/(2\hbar)}\!\!\Dff t\,\frac{1-\cos t}{(t-2n\pi)(t-2m\pi)}.
\label{eq:Amn}
\end{equation}
Maximising over the unit-norm Fourier coefficients $C$ gives
$\max_{\sket{\psi}}P_W=\opnorm{A}/\pi$. By the Slepian--Pollak relation,
$\opnorm{A}/\pi=\lambda_{0}(LW/4\hbar)$, so the existence of a state with
interval confidences $\theta_x\le 1$ at position support $X=[-L/2,L/2]$ and
$\theta_p$ at momentum support $P=[-W/2,W/2]$ is constrained by
\begin{equation}
\theta_p\;\le\;\lambda_{0}\!\left(\frac{LW}{4\hbar}\right),\quad\text{when}\;\theta_x=1.
\label{eq:Slepian_constraint}
\end{equation}
With the symmetric matrix $A$ in \eqref{eq:Amn}, we record an elementary bound
\begin{equation}
\opnorm{A}\;\ge\;\frac{\pi}{\theta_{x}}\!\left[\theta_{p}-2\sqrt{(\theta_{x}+\theta_{p}-1)(1-\theta_{x})}\right],
\label{eq:earlier_bound}
\end{equation}
in Appendix~\ref{app:earlier}, which is valid only when $2\theta_{x}+\theta_{p}>2$.
For general $\theta_x<1$ Theorem~\ref{thm:LP} gives the corresponding
inequality \eqref{eq:LP_bound_I}.

\section{Numerical landscape}\label{sec:numerics}

We now compute the bounds in Theorem~\ref{thm:LP} explicitly. The key
ingredient is the largest Slepian eigenvalue $\lambda_{0}(c)$ which we
evaluate by Gauss--Legendre discretisation of the integral operator $K_c$
(see Appendix~\ref{app:numerical} and Table~\ref{tab:lambda}).

\subsection{Eigenvalue $\lambda_{0}(c)$}

\begin{figure}[!htbp]
\centering
\includegraphics[width=\columnwidth]{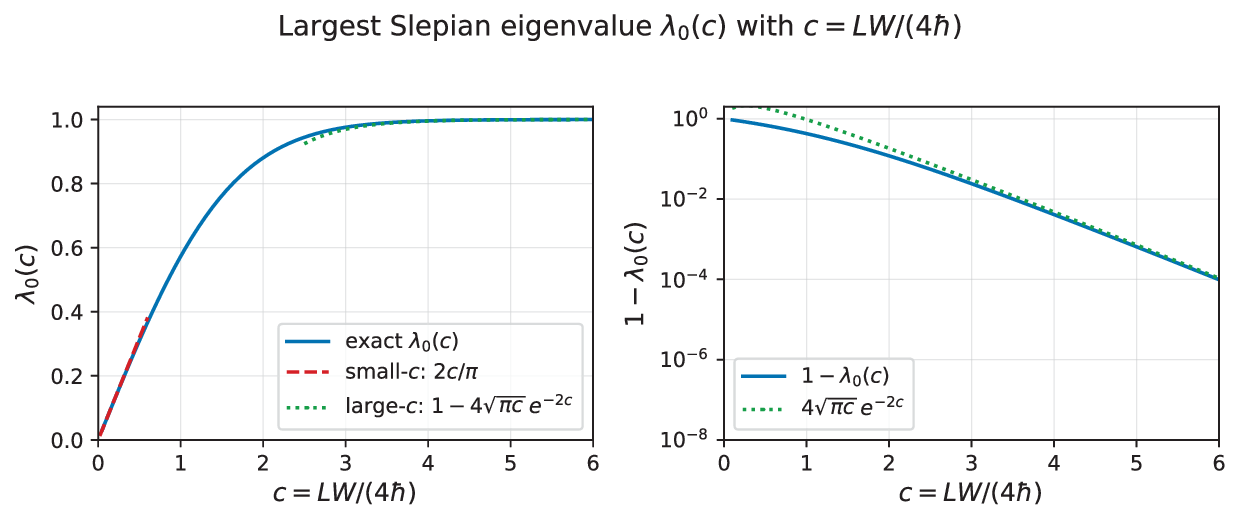}
\caption{Left: the largest Slepian eigenvalue $\lambda_{0}(c)$ on a linear
scale, with the small-$c$ asymptote $2c/\pi$ and the large-$c$ asymptote
$1-4\sqrt{\pi c}\,e^{-2c}$ \cite{Slepian1983}. Right: $1-\lambda_{0}(c)$ on a
log scale showing super-exponential decay. The variable $c=LW/(4\hbar)$
parametrises the time-bandwidth product.}
\label{fig:lambda0}
\end{figure}

\begin{table}[!htbp]
\caption{Selected values of $\lambda_{0}(c)$, computed by 400-node
Gauss--Legendre quadrature of $K_c$.}
\label{tab:lambda}
\centering
\renewcommand{\arraystretch}{1.05}
\begin{tabular}{r r r}
\toprule
$c=LW/(4\hbar)$ & $\lambda_{0}(c)$ & $1-\lambda_{0}(c)$ \\
\midrule
$0.25$  & $0.158$    & $8.42\times10^{-1}$ \\
$0.50$  & $0.310$    & $6.90\times10^{-1}$ \\
$1.00$  & $0.573$    & $4.27\times10^{-1}$ \\
$1.50$  & $0.763$    & $2.37\times10^{-1}$ \\
$2.00$  & $0.881$    & $1.19\times10^{-1}$ \\
$3.00$  & $0.976$    & $2.42\times10^{-2}$ \\
$4.00$  & $0.996$    & $4.11\times10^{-3}$ \\
$5.00$  & $0.9994$   & $6.48\times10^{-4}$ \\
$6.00$  & $0.99990$  & $9.81\times10^{-5}$ \\
$8.00$  & $0.999998$ & $2.13\times10^{-6}$ \\
$10.0$  & $1-4.4\times10^{-8}$ & $4.41\times10^{-8}$ \\
\bottomrule
\end{tabular}
\end{table}

The two asymptotic expansions
\begin{align}
\lambda_{0}(c)&\simeq\frac{2c}{\pi},\qquad c\to0^{+},\label{eq:smallc}\\
1-\lambda_{0}(c)&\simeq 4\sqrt{\pi c}\,e^{-2c},\qquad c\to\infty,\label{eq:largec}
\end{align}
are reported in~\cite{Slepian1983,LandauPollak1961} and are confirmed by our
numerics (Fig.~\ref{fig:lambda0}). They allow simple closed-form
approximations of the inverse $\lambda_0^{-1}(\theta)$ at the extremes:
\begin{equation}
\lambda_0^{-1}(\theta)\simeq\frac{\pi\theta}{2}\;(\theta\!\to\!0),\qquad
\lambda_0^{-1}(\theta)\simeq-\tfrac{1}{2}\ln(1-\theta)\;(\theta\!\to\!1).\label{eq:invasymp}
\end{equation}

\subsection{Tight bound on $\Delta^{I}x\cdot\Delta^{I}p$}

\begin{figure}[!htbp]
\centering
\includegraphics[width=\columnwidth]{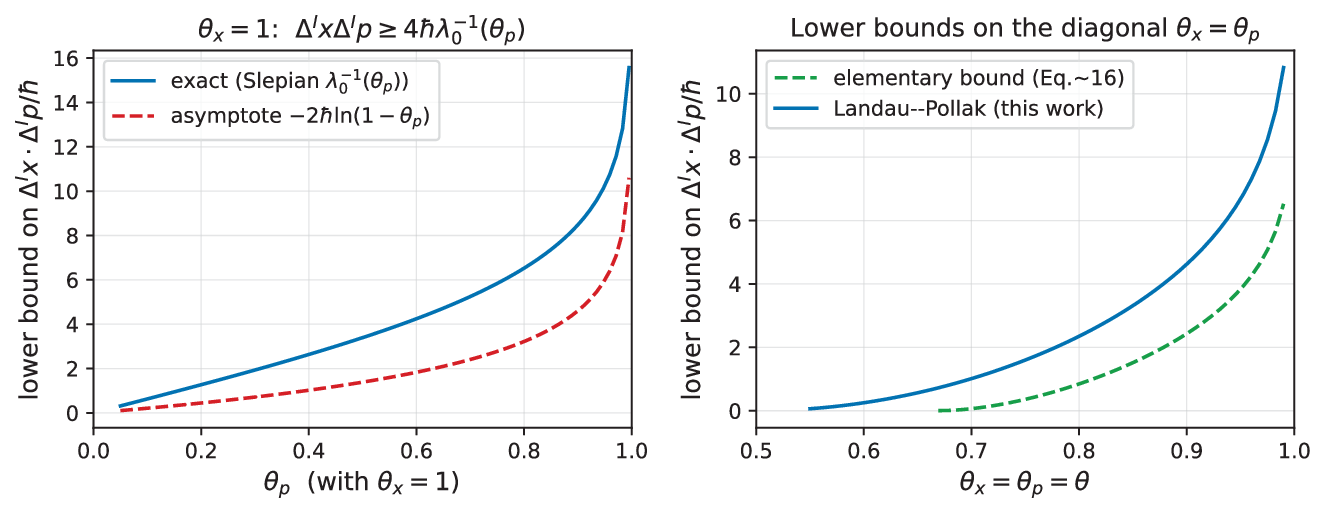}
\caption{Lower bounds on $\Delta^{I}x\,\Delta^{I}p/\hbar$. Left: along
$\theta_x=1$ the implicit Slepian bound
$4\hbar\,\lambda_{0}^{-1}(\theta_p)$ together with the closed-form
high-confidence asymptote $-2\hbar\ln(1-\theta_p)$
[Corollary~\ref{cor:logdiv}]. Right: along the diagonal
$\theta_x=\theta_p=\theta$, the new Landau--Pollak bound
strictly dominates the elementary bound, which is restricted to
$2\theta_x+\theta_p>2$ and degrades sharply when $\theta_x$ recedes from
unity.}
\label{fig:bcomp}
\end{figure}

\begin{table}[!htbp]
\caption{Tight Landau--Pollak lower bound on $\Delta^{I}x\,\Delta^{I}p/\hbar$
at sample $(\theta_x,\theta_p)$. The bound equals four times the inverse
Slepian eigenvalue of the target
$T(\theta_x,\theta_p)=(\sqrt{\theta_x\theta_p}-\sqrt{(1-\theta_x)(1-\theta_p)})^{2}$.}
\label{tab:LP}
\centering
\renewcommand{\arraystretch}{1.05}
\begin{tabular}{c c c c}
\toprule
$\theta_x$ & $\theta_p$ & $T$ & $\Delta^{I}x\,\Delta^{I}p/\hbar$ \\
\midrule
$0.60$ & $0.60$ & $0.040$ & $0.25$ \\
$0.70$ & $0.70$ & $0.160$ & $1.01$ \\
$0.80$ & $0.80$ & $0.360$ & $2.35$ \\
$0.90$ & $0.90$ & $0.640$ & $4.62$ \\
$0.95$ & $0.95$ & $0.810$ & $6.68$ \\
$0.99$ & $0.99$ & $0.9604$ & $10.82$ \\
$0.99$ & $0.50$ & $0.4005$ & $2.64$ \\
$0.95$ & $0.70$ & $0.4803$ & $3.24$ \\
$1.00$ & $0.95$ & $0.9500$ & $10.25$ \\
\bottomrule
\end{tabular}
\end{table}

Figure~\ref{fig:bcomp} compares the Landau--Pollak bound of
Theorem~\ref{thm:LP} against the elementary bound \eqref{eq:earlier_bound}. Both bounds
agree on $\theta_x=1$, but the new one extends to all
$\theta_x+\theta_p>1$ and is consistently strictly tighter on the
diagonal. Table~\ref{tab:LP} samples the tight bound across the
$(\theta_x,\theta_p)$-plane; Fig.~\ref{fig:map} displays the resulting
landscape.

\begin{figure}[!htbp]
\centering
\includegraphics[width=\columnwidth]{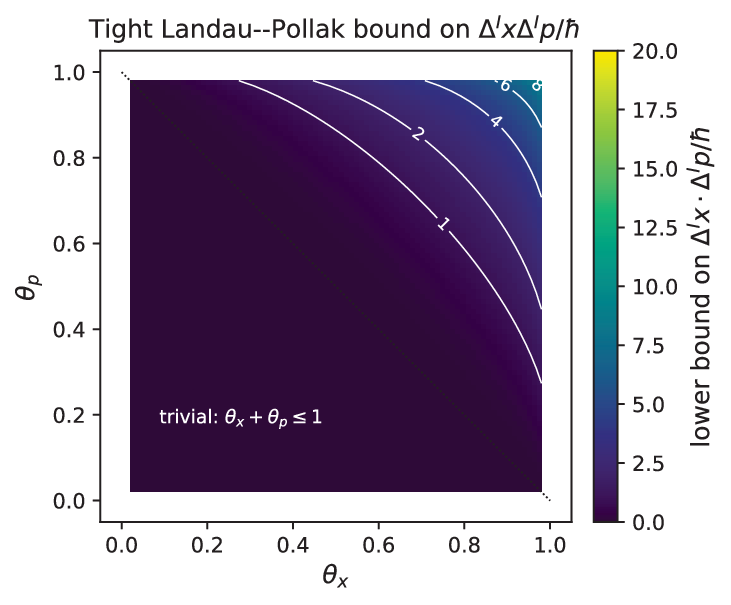}
\caption{Numerical landscape of the tight bound on
$\Delta^{I}x\,\Delta^{I}p/\hbar$ over $(\theta_x,\theta_p)\in[0,1]^{2}$.
The bound diverges as both confidences approach unity; the dashed line
$\theta_x+\theta_p=1$ separates the trivial region (Theorem~\ref{thm:half}).}
\label{fig:map}
\end{figure}

\subsection{Saturating Slepian-superposition state}

\begin{figure}[!htbp]
\centering
\includegraphics[width=0.95\columnwidth]{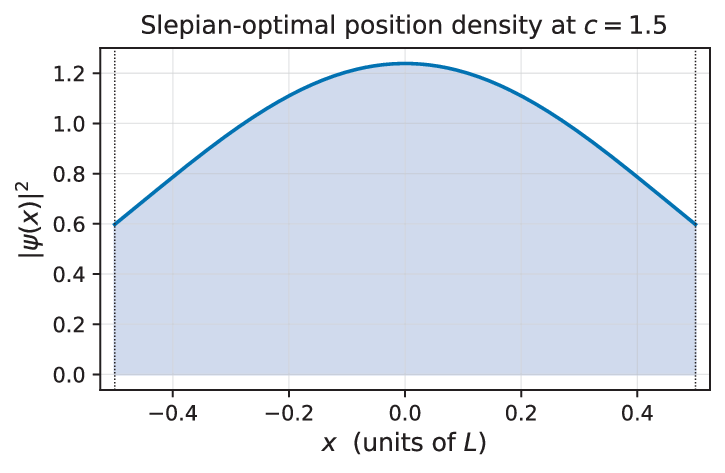}
\\[2pt]
\includegraphics[width=0.95\columnwidth]{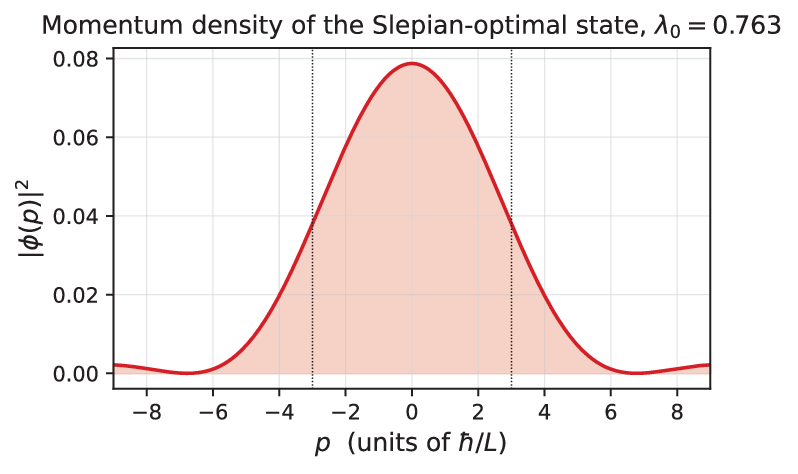}
\caption{Position (top) and momentum (bottom) probability densities of the
state that saturates Theorem~\ref{thm:LP} at $c=LW/(4\hbar)=1.5$, where
$\lambda_0(c)=0.763$. The position density vanishes outside $[-L/2,L/2]$,
while the momentum density concentrates inside $[-W/2,W/2]$ marked by the
dotted lines.}
\label{fig:slepian_states}
\end{figure}

The state saturating \eqref{eq:LP_bound_I_explicit} is the leading
prolate-spheroidal eigenfunction $\psi_{0}^{\mathrm{S}}$. We compute it on
$[-1,1]$ by the same Gauss--Legendre discretisation that supplied
$\lambda_0(c)$, and exhibit its position density together with the Fourier
transform in Fig.~\ref{fig:slepian_states} for $c=1.5$. The momentum density
is sharply concentrated within $[-W/2,W/2]$ as expected; the
$1-\lambda_{0}(c)$ tail mass leaks outside.

\section{Comparison with established uncertainty principles}\label{sec:compare}

We now place Theorem~\ref{thm:LP} in the broader landscape of uncertainty
relations.

\subsection{Heisenberg--Kennard variance bound}

The Heisenberg--Kennard bound \eqref{eq:HK} controls the variance, not the
confidence-level support. For a minimum-uncertainty Gaussian state with
$\sigma_{x}\sigma_{p}=\hbar/2$, the symmetric interval at confidence
$\theta$ has length $\Delta^{I}x(\theta)=2\sqrt{2}\,\sigma_x\,
\mathrm{erf}^{-1}(\theta)$, so that
$\Delta^{I}x(\theta)\Delta^{I}p(\theta)=4\,\mathrm{erf}^{-1}(\theta)^{2}\,\hbar$.
Comparing with the Slepian-saturating state at $\theta_x=\theta_p=\theta$
gives Fig.~\ref{fig:gaussvslepian} and Table~\ref{tab:gauss_v_slepian}: at
high confidence the Slepian state is strictly better, with the ratio
diverging as $\theta\to 1$.

\begin{figure}[!htbp]
\centering
\includegraphics[width=\columnwidth]{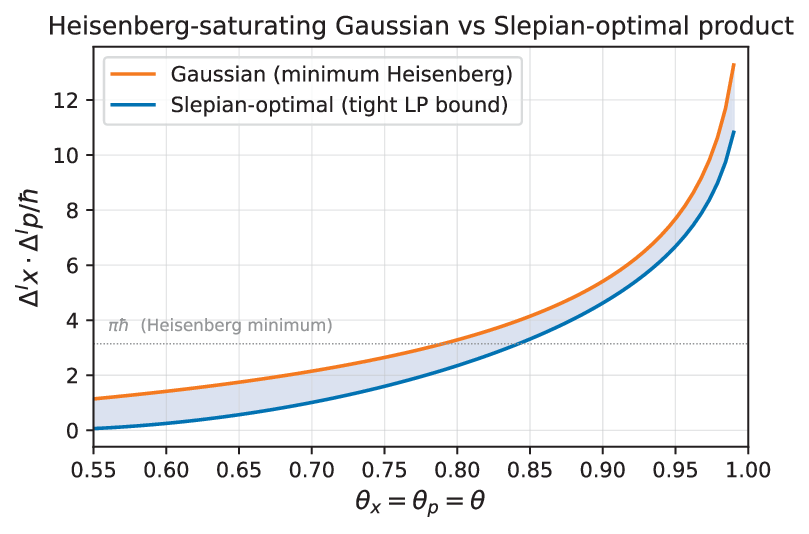}
\caption{Interval-product comparison along $\theta_x=\theta_p=\theta$. The
minimum-uncertainty Gaussian (orange) saturates Heisenberg--Kennard but is
strictly suboptimal for confidence uncertainty when $\theta>1/2$. The
Slepian-superposition state (blue) attains the tight Landau--Pollak bound
of Theorem~\ref{thm:LP}.}
\label{fig:gaussvslepian}
\end{figure}

\begin{table}[!htbp]
\caption{Interval-uncertainty product
$\Delta^{I}x\,\Delta^{I}p/\hbar$ along $\theta_x=\theta_p=\theta$ for the
Heisenberg-saturating Gaussian and for the Slepian-saturating state. The
``ratio'' column is Gaussian / Slepian; values $>1$ mean the Slepian state
has the smaller product.}
\label{tab:gauss_v_slepian}
\centering
\renewcommand{\arraystretch}{1.05}
\begin{tabular}{c c c c}
\toprule
$\theta$ & Gaussian & Slepian & ratio \\
\midrule
$0.55$ & $1.14$ & $0.063$ & $18.16$ \\
$0.60$ & $1.42$ & $0.251$ & $5.63$ \\
$0.70$ & $2.15$ & $1.013$ & $2.12$ \\
$0.80$ & $3.28$ & $2.349$ & $1.40$ \\
$0.90$ & $5.41$ & $4.622$ & $1.17$ \\
$0.95$ & $7.68$ & $6.679$ & $1.15$ \\
$0.99$ & $13.27$ & $10.82$ & $1.23$ \\
\bottomrule
\end{tabular}
\end{table}

The Slepian state strictly dominates the minimum-variance Gaussian across the
entire confidence range. The improvement is most dramatic at low confidence
($\theta\sim 0.55$) where the Slepian state achieves
$\Delta^{I}x\,\Delta^{I}p\sim 0.06\hbar$, an order of magnitude below
$\hbar/2$. This is not in conflict with Heisenberg--Kennard: the variance
$\sigma_x\sigma_p$ of the Slepian state is large (its compact position
support carries broad momentum tails), while its central interval product is
small. Variance and confidence are genuinely different uncertainty notions,
the quantitative counterpart of the standard observation that
Heisenberg--Kennard fails to control sharp bimodal
distributions~\cite{Coles2017}.

\subsection{Entropic uncertainty}

The Bia\l{}ynicki-Birula--Mycielski (BBM)
inequality~\cite{BialynickiBirula1975,Beckner1975} reads
\begin{equation}
h(x)+h(p)\ge\log(\pi e\hbar),\label{eq:BBM}
\end{equation}
with $h(x)=-\int|\psi(x)|^{2}\log|\psi(x)|^{2}\,\Dff x$ the differential
entropy. The BBM bound and the discrete Maassen--Uffink
relation~\cite{Maassen1988} control the differential or Shannon entropy of
the distributions, which is a function of the full distribution rather than
of the probability mass concentrated in a small set. Concretely, a state
that places probability $\theta$ on a thin spike of width $\varepsilon$ and
the remaining probability on a wide smooth tail can simultaneously have
$h(x)$ large (driven by the tail) and $\Delta^{c}x(\theta)\sim\varepsilon$
small. Theorem~\ref{thm:LP} therefore captures information that the entropic
relations do not, especially in the high-confidence regime $\theta\to 1$
where it gives the explicit logarithmic divergence \eqref{eq:asymp_log}.
Conversely, BBM gives a bound on the Shannon entropies that
Theorem~\ref{thm:LP} cannot reproduce. The two relations are
complementary, not comparable.

\subsection{Donoho--Stark concentration}

Donoho and Stark~\cite{DonohoStark1989} prove that for any non-zero $\psi$
$\epsilon_{T}$-concentrated on a set $T$ in position and
$\epsilon_{W}$-concentrated on a set $W$ in momentum,
\begin{equation}
|T|\,|W|\;\ge\;2\pi\hbar\bigl(1-\epsilon_{T}-\epsilon_{W}\bigr)^{2}_{+}.\label{eq:DS}
\end{equation}
Re-expressing the concentration parameters in our notation
($\epsilon_{T}^{2}=1-\theta_{x}$, $\epsilon_{W}^{2}=1-\theta_{p}$),
\eqref{eq:DS} becomes
\begin{equation}
\Delta^{c}x\,\Delta^{c}p\;\ge\;2\pi\hbar\bigl(1-\sqrt{1-\theta_{x}}-\sqrt{1-\theta_{p}}\bigr)^{2}_{+}.
\label{eq:DSour}
\end{equation}
Our bound \eqref{eq:LP_bound} replaces the linear combination
$1-\sqrt{1-\theta_x}-\sqrt{1-\theta_p}$ by the form
$\sqrt{\theta_x\theta_p}-\sqrt{(1-\theta_x)(1-\theta_p)}$ which is strictly
larger throughout $\theta_x+\theta_p>1$. For instance at
$\theta_x=\theta_p=0.9$:
\eqref{eq:DSour} gives $\Delta^{c}x\,\Delta^{c}p\ge 0.85\hbar$, while
\eqref{eq:LP_bound} gives $\Delta^{c}x\,\Delta^{c}p\ge 4.02\hbar$, a
$4.7\!\times$ improvement.

\subsection{Landau--Pollak inequality and its operator form}

For a single position interval and a single momentum interval, the
Landau--Pollak inequality~\cite{LandauPollak1961} gives the angular form
\begin{equation}
\arccos\sqrt{\theta_{x}}+\arccos\sqrt{\theta_{p}}\;\ge\;\arccos\!\sqrt{\lambda_{0}\!\left(\frac{|T||W|}{4\hbar}\right)},\label{eq:LPangular}
\end{equation}
which is exactly \eqref{eq:LP_bound_I}. The operator-theoretic version of
the same inequality is Lenard's two-projection
theorem~\cite{Lenard1972}. Theorem~\ref{thm:LP} extracts the position-momentum
content of these inequalities, restated in the language of confidence
uncertainty, and supplies the explicit
high-confidence asymptote \eqref{eq:asymp_log} that the operator-theoretic
form does not directly disclose.

\section{Discussion and limitations}\label{sec:disc}

We summarise the qualitative picture, then list open questions.

\paragraph*{Picture.}
For $(\theta_x,\theta_p)$ \emph{below} the antidiagonal we can achieve
arbitrarily small $\Delta^{c}x$ and $\Delta^{c}p$ simultaneously
(Theorem~\ref{thm:half}); above the antidiagonal a strict positive lower
bound holds (Theorem~\ref{thm:LP}), saturated by Slepian states for the
interval version. Our numerical evaluation of $\lambda_0(c)$ converts the
implicit bound into explicit values across the entire upper triangle
(Sec.~\ref{sec:numerics}); the bound is strictly tighter than the elementary
bound \eqref{eq:earlier_bound} and dominates the Donoho--Stark concentration
bound by factors of $4$--$5$ at $\theta_{x},\theta_{p}\sim 0.9$.

\paragraph*{Tightness.}
\eqref{eq:LP_bound_I} is tight: it is saturated by the principal
prolate-spheroidal eigenfunction. In the union-of-intervals confidence
setting \eqref{eq:LP_bound} we use the Donoho--Stark Hilbert--Schmidt bound
$\opnorm{P_{X}Q_{P}}^{2}\le\frac{|X||P|}{2\pi\hbar}$, which is generally not
tight for unions of intervals; closing the gap to a tight bound is an open
problem (cf.~\cite{Nazarov2008}).

\paragraph*{Mixed states.}
For mixed states $\rho$ the same projection-and-Hilbert-Schmidt argument
applies provided $\theta_{\bullet}=\mathrm{tr}(P_{\bullet}\rho)$, leaving
\eqref{eq:LP_bound} valid as stated. The interval bound
\eqref{eq:LP_bound_I_explicit} also extends to mixed states; the saturating
state is pure (the principal prolate-spheroidal eigenfunction).

\paragraph*{Multidimensional generalisation.}
The Slepian theory extends to $\mathbb{R}^{d}$ via the multidimensional
prolate-spheroidal wave functions~\cite{Slepian1983,Osipov2013}, opening a
direct $d$-dimensional analogue of Theorem~\ref{thm:LP}. Time-energy and
number-phase confidence uncertainties can be formulated similarly using the
appropriate one-parameter unitary group.

\paragraph*{Experimental relevance.}
The $50\%$ joint-localisation result of Theorem~\ref{thm:half} is
counter-intuitive but consistent with the standard quantum
formalism: it reflects the freedom to design states with two narrow
support sets, one in position and one in momentum, sharing the same
overall norm. Realising such states experimentally is an open
question; the Slepian-superposition states required for the saturation of
\eqref{eq:LP_bound_I} are accessible numerically and ought to be amenable
to digital quantum simulation.

\section{Conclusion}\label{sec:conclu}

We have unified the two complementary uncertainty notions of confidence
uncertainty $\Delta^{c}$ and interval confidence uncertainty $\Delta^{I}$
and established the complete two-region phase diagram on
$(\theta_x,\theta_p)\in[0,1]^{2}$. Below the antidiagonal,
$\theta_x+\theta_p\le 1$, position and momentum can be localised
arbitrarily well simultaneously; above it the new bound
$\Delta^{c}x\,\Delta^{c}p\ge2\pi\hbar(\sqrt{\theta_x\theta_p}-\sqrt{(1-\theta_x)(1-\theta_p)})^2$
holds, with a sharp Slepian-eigenvalue version for the interval form. The
numerical evaluation of $\lambda_{0}(c)$ converts these implicit statements
into explicit, quantitative bounds, and the comparison with the
Heisenberg--Kennard, Bia\l{}ynicki-Birula--Mycielski, and Donoho--Stark
inequalities clarifies in which regimes the new bound is tighter. We
expect the framework to be useful for experimental questions in which what
matters is the probability concentrated on a small support, rather than the
variance of the distribution.

\begin{acknowledgments}
We gratefully acknowledge Sixia Yu for valuable discussions.
This work is supported by the National Natural Science Foundation of China
(Grant No.~12475020), the National Key Research and
Development Program of China (Grant No.~2023YFC2205802), and the Innovation
Program for Quantum Science and Technology (Grant No.~2021ZD0301701).
\end{acknowledgments}



\appendix

\section{Proof of Theorem~\ref{thm:LP}}\label{app:proofs}\label{app:LP}

Let $X,P\subseteq\mathbb{R}$ be measurable sets, $|X|<\infty$, $|P|<\infty$.
Write $P_{X}$ for the multiplication operator $f(x)\mapsto\mathbf{1}_{X}(x)f(x)$ on
$L^{2}(\mathbb{R})$, and $Q_{P}=\mathcal{F}^{-1}\mathbf{1}_{P}\mathcal{F}$ for
the corresponding spectral projection in momentum, with the Fourier
transform $\mathcal{F}\psi(p)=(2\pi\hbar)^{-1/2}\int e^{-ipx/\hbar}\psi(x)\,\Dff x$.
The kernel of $P_{X}Q_{P}$ in the position representation is
$K(x,y)=\mathbf{1}_{X}(x)\frac{1}{2\pi\hbar}\int_{P}e^{ip(x-y)/\hbar}\Dff p$,
hence
\begin{equation}
\HSnorm{P_{X}Q_{P}}^{2}=\frac{|X||P|}{2\pi\hbar},\qquad\opnorm{P_{X}Q_{P}}\le\HSnorm{P_{X}Q_{P}}.\label{eq:HS}
\end{equation}
This is the Donoho--Stark Theorem~1~\cite{DonohoStark1989}.

For any unit vector $\psi$, Lenard's projection
inequality~\cite{Lenard1972} states
\begin{equation}
\arccos\|P_{X}\psi\|+\arccos\|Q_{P}\psi\|\;\ge\;\arccos\opnorm{P_{X}Q_{P}},\label{eq:Lenard}
\end{equation}
where the right-hand side is the principal angle between the ranges of $P_{X}$
and $Q_{P}$. If $\|P_{X}\psi\|^{2}\ge\theta_{x}$ and
$\|Q_{P}\psi\|^{2}\ge\theta_{p}$ then since $\arccos$ is decreasing
\begin{equation}
\arccos\sqrt{\theta_{x}}+\arccos\sqrt{\theta_{p}}\;\ge\;\arccos\opnorm{P_{X}Q_{P}}.\label{eq:Lenard2}
\end{equation}
For $\theta_{x}+\theta_{p}>1$ both sides lie in $[0,\pi/2]$, so taking
cosine and expanding via
$\cos(\alpha+\beta)=\cos\alpha\cos\beta-\sin\alpha\sin\beta$ yields
\begin{equation}
\sqrt{\theta_{x}\theta_{p}}-\sqrt{(1-\theta_{x})(1-\theta_{p})}\;\le\;\opnorm{P_{X}Q_{P}}.\label{eq:opLB}
\end{equation}
Squaring and combining with \eqref{eq:HS},
\begin{equation}
\bigl(\sqrt{\theta_{x}\theta_{p}}-\sqrt{(1-\theta_{x})(1-\theta_{p})}\bigr)^{2}\le\frac{|X||P|}{2\pi\hbar}.\label{eq:LP_combined}
\end{equation}
Now \eqref{eq:LP_combined} holds for \emph{every} pair $(X,P)$ realising the
confidence levels. Taking the infimum over all such pairs gives
\eqref{eq:LP_bound}.

For the interval version, $X=[x_1,x_2]$ and $P=[p_1,p_2]$ are single
intervals. The Slepian--Pollak identity~\cite{SlepianPollak1961} states
\begin{equation}
\opnorm{P_{X}Q_{P}}^{2}=\lambda_{0}\!\left(\frac{|X||P|}{4\hbar}\right),\label{eq:SP}
\end{equation}
so \eqref{eq:opLB} promotes to
\begin{equation}
(\sqrt{\theta_{x}\theta_{p}}-\sqrt{(1-\theta_{x})(1-\theta_{p})})^{2}\le\lambda_{0}\!\left(\frac{|X||P|}{4\hbar}\right),
\end{equation}
which is \eqref{eq:LP_bound_I}. Saturation is obtained at the
$\psi=\sqrt{\theta_{x}}\psi_{0}^{\mathrm{S}}+\sqrt{1-\theta_{x}}\psi_{0}^{\mathrm{S},\perp}$
linear combination, where $\psi_{0}^{\mathrm{S}}$ is the leading
prolate-spheroidal eigenfunction; this is the standard Lenard-equality
configuration, see~\cite{Lenard1972} and the recent
discussion~\cite{FollandSitaram1997}.\hfill$\square$

\section{Proof of Corollary~\ref{cor:logdiv}}\label{app:cor}

For $\theta_x=1$ the angular target reduces to
$T(1,\theta_p)=\theta_p$. Hence $\Delta^{I}x\,\Delta^{I}p/\hbar\ge
4\,\lambda_{0}^{-1}(\theta_p)$. By \eqref{eq:largec}
$1-\lambda_{0}(c)\sim 4\sqrt{\pi c}\,e^{-2c}$ as $c\to\infty$, so
$\ln(1-\theta_p)\sim\ln(4\sqrt{\pi c})-2c$, equivalently
$c\sim-\tfrac{1}{2}\ln(1-\theta_p)+\tfrac{1}{4}\ln(\pi c)+\ln 2$, of which
the leading term is $c\sim-\tfrac{1}{2}\ln(1-\theta_p)$. Multiplying by $4$
gives \eqref{eq:asymp_log}.\hfill$\square$

\section{Numerical recipe for $\lambda_{0}(c)$}\label{app:numerical}

We discretise the Slepian operator $K_c$ on $L^2[-1,1]$ by
Gauss--Legendre quadrature: pick nodes $\{u_i\}_{i=1}^{N}$ and weights
$\{w_i\}_{i=1}^{N}$, build the symmetric matrix
\begin{equation}
M_{ij}=\sqrt{w_iw_j}\,\frac{\sin c(u_i-u_j)}{\pi(u_i-u_j)}\quad(\text{taking }c/\pi\text{ at }u_i=u_j),
\end{equation}
and read off $\lambda_{0}(c)$ as the largest eigenvalue of $M$. Convergence
is spectral in $N$; we use $N=400$, comfortably resolving $1-\lambda_{0}(c)$
down to $10^{-8}$. The leading eigenvector returns the principal Slepian
function $\psi_0^{\mathrm{S}}$ on the quadrature nodes, which we use to
generate Fig.~\ref{fig:slepian_states}.

\section{An elementary bound implied by Theorem~\ref{thm:LP}}\label{app:earlier}

For completeness we record an elementary bound which can be derived without
appeal to Lenard's inequality. It is the bound that appeared in our earlier
analysis and is a special case of \eqref{eq:LP_bound_I} valid only for
$2\theta_{x}+\theta_{p}>2$. Decompose the state, with respect to a candidate
position interval $[-L/2,L/2]$ that captures probability $\theta_{x}$, as
\begin{equation}
\psi=\sqrt{p_{-}}\psi_{-}+\sqrt{\theta_{x}}\psi_{0}+\sqrt{p_{+}}\psi_{+},
\end{equation}
where $\psi_{0}$ is supported in $[-L/2,L/2]$, $\psi_{+}$ in
$[L/2,L_{+}+L/2]$, $\psi_{-}$ in $[-L_{-}-L/2,-L/2]$, all normalised, and
$p_{-}+p_{+}+\theta_{x}=1$. Writing the momentum probability in the band
$[-W/2,W/2]$ as $P_{W}=\int_{-W/2}^{W/2}|\phi|^{2}\,\Dff p$, expanding the
cross-terms by Cauchy--Schwarz, and using
$P_{W}(\psi_{0})\le\opnorm{A(LW)}/\pi$ together with
$P_{W}(\psi_{\pm})\le\opnorm{A(L_{\pm}W)}/\pi\le 1$ gives
\begin{equation}
\begin{aligned}
P_{W}(\psi)&\le\theta_{x}\frac{\opnorm{A(LW)}}{\pi}+p_{+}+p_{-}\\
&\quad+2\sqrt{\theta_{x}p_{+}}\sqrt{\frac{\opnorm{A(LW)}}{\pi}}+\ldots
\end{aligned}
\end{equation}
Imposing $P_{W}(\psi)\ge\theta_{p}$ and $p_{+}+p_{-}=1-\theta_{x}$ and
optimising in $p_{\pm}$ gives, when $2\theta_{x}+\theta_{p}>2$,
\begin{equation}
\opnorm{A}\Big|_{LW=\Delta^{I}\!x\Delta^{I}\!p}\;\ge\;\frac{\pi}{\theta_{x}}\!\left[\theta_{p}-2\sqrt{(\theta_{x}+\theta_{p}-1)(1-\theta_{x})}\right].\label{eq:elementary}
\end{equation}
The Landau--Pollak result \eqref{eq:LP_bound_I} is strictly tighter on its
domain $\theta_{x}+\theta_{p}>1$, the two bounds agreeing only on the
boundary $\theta_{x}=1$.\hfill$\square$

\end{document}